\documentclass[12pt]{amsart}
\usepackage{amsmath, amssymb, amsthm, amsfonts, mathrsfs}
\usepackage{cite}
\usepackage{amscd}
\usepackage{url}

\newtheorem{prop}{Proposition}

\newtheorem{theorem}[prop]{Theorem}

\theoremstyle{definition}
\newtheorem{example}[prop]{Example}

\newcommand{\R}{\mathbb{R}}
\newcommand{\ess}[1]{\hat{#1}}

\newcommand{\ex}[2]{\mathbb{E}_{#1}\left[#2\right]}

\begin{document}
\title{Information Geometry and Evolutionary Game Theory}
\author{Marc Harper}
\address{University of California Los Angeles}
\email{marcharper@ucla.edu} 
% \authorinfo{\texttt{c}}
\date{\today}

\begin{abstract}
The Shahshahani geometry of evolutionary game theory is realized as the information geometry of the simplex, deriving from the Fisher information metric of the manifold of categorical probability distributions. Some essential concepts in evolutionary game theory are realized information-theoretically. Results are extended to the Lotka-Volterra equation and to multiple population systems.
\end{abstract}

\maketitle

\section{Introduction}

The replicator equation is a widely-used model of natural selection. This paper explains the realization of the geometry of evolutionary game theory in terms of information theoretical principles, giving a purely mathematical and statistical origination of the replicator equation. Under this interpretation, the replicator equation models the information dynamics of a population of replicating entities. Additionally, the Kullback-Liebler information divergence, which serves as a Lyapunov function for the replicator dynamic, can be interpreted as a measure potential information, characterizing the concept of evolutionary stable state informatically.

\subsection{Continuous Replicator Dynamic}

Consider a categorical distribution $X$ on $n$ categories of entities in a population. This is discrete probability distribution represented by a unit vector of $n$ variables $x = (x_1, \ldots, x_n)$ under the normalization $|x| = x_1 + \cdots + x_n = 1$, where $x_i$ denotes the proportion of the $i$-th type in the population. The replicator equation on this distribution is the differential equation
\[ \dot{x}_i = x_i\left(f_i(x) - \bar{f}(x)\right),\]
where $f(x) = (f_1(x), \ldots, f_n(x))$ is a fitness landscape and $\bar{f}(x)= x_1 f_1(x) + \cdots + x_n f_n(x)$ is the mean fitness.

\section{Geometric Aspects of Evolutionary Game Theory}

The information theoretic interpretations of the previous chapter have a unified basis in information geometry. We begin with a description of the geometry of the simplex and geometric results known in evolutionary game theory.

\subsection{The Geometry of the Simplex}

Let $S^n$ be the interior of the $n$-simplex $\Delta^n$, which is $(n-1)$-dimensional. Each point $x$ of the simplex has the property that $x_1 + \cdots + x_n = 1$, so the tangent space at any point on the interior is the $(n-1)$-dimensional vector space described by $n$ vectors $v_1, \ldots, v_n$ such that $v_1 + \cdots + v_n = 0$. The orthogonal complement of the tangent space is the one dimensional line with direction vector $\mathbf{1} = (1, 1, \ldots, 1)$. Indeed $\mathbf{1} \cdot v = 0$ for any $v$ in the tangent space and the complement is necessarily one-dimensional. The faces of $\Delta^n$ are isomorphic to a simplex of one lower dimension, which can be seen by setting one of the $x_i$ to zero and indicates the absence of that type in the population. The replicator equation is \emph{forward-invariant} on the simplex (and hence each of its faces), since if $x_i = 0$ then $\dot{x}_i = 0$. Because of this property, the replicator equation is called \emph{non-innovative} since new types cannot arise, in contrast to evolutionary dynamics in which this is possible (notably the replicator-mutator equation \cite{Nowak06} and the orthogonal projection dynamic \cite{Sandholm09}).

\subsection{Shahshahani Geometry}

Shahshahani introduced two Riemannian manifolds into mathematical biology\cite{Shahshahani79}: the positive orthant of $\R^n$, denoted $\R^n_+$, with the metric \[ g_{ij}(x) = \frac{|x|}{x_i}\delta_{ij}, \] where $|x| = \sum_{i}{x_i}$ and the restriction to the simplex $\Delta^n = \{ x \in \R^n \, | \, |x| = 1, x_i \geq 0 \, \forall i \} \subset \R^n_+$, with the metric \[ g_{ij}(x) = \frac{1}{x_i}\delta_{ij}. \]
Call the latter manifold the Shahshahani manifold; its metric is known as the Shahshahani metric. There is a normalization map $N: \R^n_{+} \to \Delta^n$ given by $x \mapsto \frac{x}{|x|}$. For each $\tau \in \R_{+}$, there is a map $\varphi_{\tau}$ mapping the simplex into $\R^n_{+}$ by $x \mapsto \tau x$. These maps are sections of the normalization map since $N \circ \varphi_{\tau} = \text{id}_{\Delta^n}$. The Shahshahani metric diverges on the boundary of the simplex so the metric is valid only on the interior $S^n$. Dynamics that are forward-invariant, such as the replicator dynamic, are not affected by the discontinuity at the boundary.

%  (Shahshahani considered a variant of the replicator equation that included chromosomal crossover, which is potentially innovative.) 
% The Shahshahani manifold can be transformed into the positive portion of the $n$-sphere of radius 2 with the Euclidean metric by the mapping $x_i \to 2\sqrt{x_i}$\cite{Akin79}. This gives a formula for the geodesic distance
% \[ d(p,q) = 2 \arccos{ \left( \sum_{i}{ \sqrt{p_i q_i} } \right)}. \]
% This equation has been proposed as a measure of genetic distance\cite{Burger}.

\subsection{The Replicator Dynamic, Geometrically}

The geometry of the Shahshahani manifold yields an elegant interpretation of the replicator equation: it is the gradient flow of the Shahshahani metric. Shahshahani proved the result for a special case of the replicator equation; the following more general theorem comes from \cite{Hofbauer98}.
\begin{theorem}
If the differential equation $\dot{x_i} = f_i(x)$ is a Euclidean gradient with $f_i = \frac{\partial V}{\partial x_i}$ then the replicator equation $x_i = \ess{f}_i(x) = x_i(f_i(x) - \bar{f}(x))$ is a gradient with respect to the Shahshahani metric.
\end{theorem}

% \begin{proof}
% To show that the replicator equation is a Shahshahani gradient we use the fact that the gradient is uniquely defined as the vector representing the directional derivative. Let $z$ be a vector in the tangent space of the simplex at the point $x$ and so necessarily $\sum_{i}{z_i} = 0$. Let $V$ be the Euclidean potential of $f$ so that $\ds{\frac{\partial V}{\partial x_i}(x) = f_i(x) = \dot{x}_i}$. Using the inner product on the tangent bundle,
% \begin{align*}
% <\ess{f}(x), z>_x &= \sum_{i = 1}^{n}{\frac{1}{x_i} \ess{f}_i(x) z_i} = \sum_{i = 1}^{n}{ f_i(x) z_i} - \bar{f}(x)\sum_{i=1}^{n}z_i\\
% &= \sum_{i=1}^{n}{ \frac{\partial V}{\partial x_i}(x) z_i} = D_x V(z)
% \end{align*}
% \end{proof}

In the case that the fitness landscape is a Euclidean gradient, the Shahshahani gradient gives a Lyapunov function for the dynamic. The classical case is that of a symmetric matrix $A$ and fitness landscape $f(x) = Ax$, where $A$ is the matrix of Malthusian fitness parameters given by the difference in birth rates and death rates $a_{ij} = b_{ij} - d_{ij}$ of an individual having alleles $i$ and $j$, where the alleles are of a single gene locus. In this case the Shahshahani potential is the mean fitness $\frac{1}{2} x \cdot f(x) = \frac{1}{2} x \cdot Ax$, with $Ax$ the Euclidean gradient\cite{Shahshahani79, Hofbauer98}.

\subsection{Fisher's Fundamental Theorem and Kimura's Maximal Principle}

Fisher's fundamental theorem is a consequence of the geometric approach.
\begin{theorem}\label{fft_egt}
The rate of change of the Shahshahani potential is equal to the variance of the fitness landscape \cite{Hofbauer98}.
\[ \dot{V}(x) = \text{Var}_{x}[f(x)].\]
\end{theorem}
% \begin{proof}
% \begin{align*}
% \dot{V}(x) &= D_x V(\dot{x}) = <\ess{f}(x), \dot{x}> = <\ess{f}(x), \ess{f}(x)>\\
% &= \sum_{i=1}^{n}{ \frac{1}{x_i} \left[x_i (f_i(x) - \bar{f}(x))\right]^2}\\
% &= \sum_{i=1}^{n}{x_i (f_i(x) - \bar{f}(x))^2} = \text{Var}_{x}[f(x)].
% \end{align*}
% \end{proof}

This is a general version of Fisher's Fundamental Theorem of Natural Selection, specializing to the traditional result in the case of a symmetric and linear fitness landscape\cite{Hofbauer98}. Kimura's maximal principle follows from the fact that the replicator equation is a gradient flow\cite{Shahshahani79}. As these are both important results in mathematical biology emerging from the geometry, an interpretation is desired of the Shahshahani metric that provides intuition for the introduced geometry on the simplex in the context of modeling natural selection.

\section{The Information Geometry of Natural Selection}

An intuitive interpretation of the Shahshahani geometry is provided by information theory. Information geometry\cite{Amari93} studies manifolds of probability distributions $p(s, x)$ on a set $S$ depending on parameters $x$, which are the coordinates of the manifold. The manifold is endowed with the Fisher information metric,
\[ g_{ij}(x) = \mathbb{E}\left[ \frac{\partial \log p}{\partial x^i} \frac{\partial \log p}{\partial x^j}  \right] \]
which can be shown to be the unique (up to a constant) metric respecting sufficient statistics \cite{Chentsov}.

\subsection{The Fisher Information Metric is the Shahshahani Metric}

The manifold of immediate interest is $P(X)$, the set of categorical probability distributions on a finite set $X$, with the Fisher information metric. In this case, it is convenient to abuse notation by allowing the parameters $x$ and distribution variables $s$ to have the same symbolic representation. \footnote{Different coordinates are sometimes chosen in information geometry for $P(X)$, letting (for instance) $x_{n+1} = 1 - \sum_{i=1}^{n}{x_i}$ to enforce $\sum_{i}{x_i} = 1$. This yields an asymmetric set of replicator equations, so a different set of coordinates is chosen in this exposition.} There is a natural mapping $\varphi: P(X) \to \Delta^n$, where $|X| = n$, given by $p \to (p(1), p(2), \ldots, p(n)) = (x_1, \ldots, x_n)$. This maps $P(X)$ isometrically onto the simplex, and an easy computation shows that that the Fisher information metric is induced by the Shahshahani metric under this mapping. Simply observe that
% \begin{align*}
\[g_{ij}(x) = \mathbb{E}\left[ \frac{\partial \log x}{\partial x^i} \frac{\partial \log x}{\partial x^j}  \right]
= \sum_{k}{x_k \frac{1}{x_i}\delta_{ik} \frac{1}{x_j}\delta_{jk}} = \frac{1}{x_i} \delta_{ij} \]
% \end{align*}

This result was recognized in \cite{Amari95} and \cite{Nihat05}.

\subsection{Fisher's Fundamental Theorem}
Fisher's fundamental theorem is built into the geometry of $P(X)$\cite{Amari93}. Define the maps $E[g]$ on $P(X)$ by $p \mapsto E_{p}[g]$, where $g$ is from the set of functions ${\R}^X = \{ g : X \to \R \}$ and $E_{p}$ is the mean taken at the distribution $p$. Similarly, let $V_p[g]$ denote the variance of the function $g$ at $p$.
\begin{theorem}
For any $g: X \to \R \in {\R}^X$, \[ V_{p} = ||(\text{d}E[g])_p||_{p}^{2} = || (\nabla E[g])_p||_{p}^{2} ,\]
where the norm is induced by the Fisher information metric, for all $p \in P(X)$.
\end{theorem}
% The proof of the theorem given in \cite{Amari93} is almost identical to the proof of Theorem \ref{fft_egt} under proper translations of notation.

% We have seen that the Fisher information metric satisfies a particular uniqueness criterion. In the next section another motivation for the Fisher information metric is given in terms of information divergences.

\subsection{Information Divergences and Metrics on $P(X)$}

Some Riemannian metrics on $P(X)$ can be derived from information divergences\cite{Amari93}. Information geometry defines an \emph{information divergence} as a smooth function $D(\cdot || \cdot): P(X) \times P(X)$ such that $D(x||y) \geq 0$ with equality iff $x = y$. The second order Taylor expansion in either variable evaluated along the diagonal $x=y$ begins with the Hessian term $H$. Indeed,
\begin{align*}
D(x||y) &= D(x || y)|_{x=y} + (\nabla D(x||y)|_{x=y})\cdot(x-y) + \frac{1}{2}\cdot(x-y)^{T}H(x)\cdot(x-y) + \cdots\\
&= 0 + 0 + \frac{1}{2}\cdot(x-y)^{T}H(x)\cdot(x-y) + \cdots\\
\end{align*}
because the gradient is parallel to $\mathbf{1}$ and $\mathbf{1} \cdot (x-y) = \mathbf{1} \cdot x - \mathbf{1} \cdot y = 1 -1 = 0$.

% In the second equation, the term $(\cdots)\textbf{1}$ is the gradient of $D$ and hence lies in the tangent space, so that the dot product with $(x-y)$ is zero.

In the case that the Hessian is positive definite, it can be used to define a metric,
\[ g_{ij}^{(D)} = \left(\frac{\partial^2 D}{\partial x_i \partial y_j} \right)|_{x=y}. \]

A metric then defines a gradient flow, hence a global information divergence yields a dynamical system on the simplex. Importantly, the Hessian of the Kullback-Liebler divergence (in either variable, evaluated on the diagonal) is the Fisher information matrix, yielding the local to global connection of these two measures of information.

\begin{example}[Kullback-Liebler Divergence]
The Kullback-Liebler divergence localizes to the Fisher information metric. In coordinates we obtain the Shahshahani metric since
\[ \frac{\partial^2}{\partial x_i \partial y_j}{D_{KL}(x || y)}|_{x=y} = \frac{1}{x_i} \delta_{ij}.\]
Hence the induced gradient flow is the replicator equation. This allows the interpretation of Fisher's Fundamental theorem and Kimura's maximal principle in terms of Fisher information: natural selection forms a gradient with respect to an informatic measure, and hence locally has the direction of maximal information increase. The rate of change of the mean fitness of the population is given by the informatic variance.
\end{example}

% \begin{remark}
% This is the second description of the replicator equation in terms of information-theoretic concepts. Locally, the replicator equation is the gradient of the Fisher information metric. Globally, the replicator equation is intimately related to the Kullback-Liebler information divergence, which is localized to the Fisher information metric.
% \end{remark}

\subsection{Kullback-Liebler Divergence is a Lyapunov function for the Replicator Dynamic}

The following theorem shows that the Kullback-Liebler information divergence forms a Lyapunov function for the replicator dynamic, given an evolutionarily stable state. In fact, evolutionary stability is characterized by this property. A version of this theorem was proved in \cite{Akin79} and in \cite{Akin90}.

\begin{theorem}\label{ess_Lyapunov}
The state $\ess{x}$ is an interior ESS for the replicator dynamic if and only if $D_{KL}(\ess{x} || x)$ is a local Lyapunov function.
\end{theorem}
\begin{proof}
Let $V(x) = D_{KL}(\ess{x} || x) = \sum_{i}{\ess{x}_i \log{\ess{x}_i}} - \sum_{i}{\ess{x}_i \log{x_i}}$.
Then we have that
\begin{align*} \dot{V}(x) &= -\sum_{i}{ \ess{x}_i \frac{\dot{x}_i}{x_i}} = -\sum_{i}{ \ess{x}_i (f_i(x) - \bar{f}(x)) }\\
&= -\sum_{i}{ \ess{x}_i f_i(x)} + \sum_{i}{ \ess{x}_i\bar{f}(x)} = -\sum_{i}{ \ess{x}_i f_i(x)} + \left(\sum_{i}{ \ess{x}_i}\right)\bar{f}(x)\\
&= -\sum_{i}{ \ess{x}_i f_i(x)} + \bar{f}(x) = -(\ess{x} \cdot f(x) - x \cdot f(x)) < 0.
\end{align*}
The last inequality holds if and only if $\ess{x}$ is an ESS. Finally, by Jensen's inequality, $D_{KL}$ is minimized when $x = \ess{x}$, so it is a local Lyapunov function.
\end{proof}

A similar result is proven in \cite{Hofbauer98}, with the Lyapunov function $V(x) = \prod_{i}{x_{i}^{ \hat{x}_i}}$, but the informatic origin is not apparent in this form, although the quantity $V$ can be interpreted as the probability of finding a categorical distribution on $x$ in the state $\hat{x}$. The quantity $D_{KL}(\ess{x} || x)$ can be described as the \emph{potential information} of the replicator system. The above result can then be interpreted information theoretically -- natural selection acts to minimize the potential information.

% This result is the continuous analog to the use of the Kullback-Liebler divergence as a measure of information gain of Bayesian inference. Within the neighborhood of the ESS, the system is minimizing the information divergence between the current population distribution and that of the selectively stable configuration. To understand these results as more than coincidence and to understand the relationship between information theory and evolutionary game theory, information geometry and an understanding of the geometric approach to evolutionary game theory is required.

Theorem \ref{ess_Lyapunov} holds for a class of ecological dynamics. A dynamic of the form $\dot{x}_i = x_i g_i(x), \, i=1,\ldots ,n$ (an \emph{ecological dynamic}) is called \emph{aggregate monotone} with respect to a fitness landscape $f$ if $g = (g_1, \ldots, g_n)$ has the property that $y \cdot f(x) > z \cdot f(x)$ if and only if $y \cdot g(x) > z \cdot g(x)$, for all distributions $x,y,z$. An aggregate monotone dynamic is the replicator dynamic up to a change in velocity \cite{Ritzberger95}. In particular, the replicator equation with a convex function applied to the fitness landscape is aggregate monotone. Consider the following extension of Theorem \ref{ess_Lyapunov}.

\begin{theorem}\label{ess_Lyapunov_aggregate}
For an aggregate monotone ecological dynamic $\dot{x}_i = x_i g_i(x)$, $D_{KL}(\ess{x} || x)$ is a Lyapunov function for the dynamic if $\ess{x}$ is an interior ESS.
\end{theorem}
\begin{proof}
Let $V(x) = D_{KL}(\ess{x} || x) = \sum_{i}{\ess{x}_i \log{\ess{x}_i}} - \sum_{i}{\ess{x}_i \log{x_i}}$.
Note that since $\dot{x}_i = x_i g_i(x)$ is a dynamic on the simplex, $0 = \sum_{i}{x_i g_i(x)} = x \cdot g(x)$.
Then we have that
\begin{align*} \dot{V}(x) &= -\sum_{i}{ \ess{x}_i \frac{\dot{x}_i}{x_i}} = -\sum_{i}{ \ess{x}_i g_i(x) }\\
&= -\ess{x} \cdot g(x) = -(\ess{x} \cdot g(x) - x \cdot g(x))
\end{align*}
Applying aggregate monotonicity to the last equality completes the proof.
\end{proof}

Since a change of velocity does not alter the orbits of the dynamic, Theorem \ref{ess_Lyapunov_aggregate} shows that the replicator equation is essentially the only aggregate monotone ecological dynamic in which evolutionary stability corresponds to minimizing the Kullback-Liebler divergence. For exactly which class of evolutionary dynamics this property holds for is an open question. From the proof it is clear that the assumption of aggregate monotonicity is too strong for a full characterization since it is only needed that $\ess{x} \cdot g(x) - x \cdot g(x) > 0$ if $\ess{x} \cdot f(x) - x \cdot f(x) > 0$, which quantifies over two distributions rather than three.

\subsection{Exponential Families as Solutions of the Replicator Equation}
The exponential map on the Shahshahani manifold is
\[ exp(x,v) = \sum_{i}{ \frac{x_i e^{v_i}}{ \sum_{j}{x_j e^{v_j}} } \hat{e_i} }, \]
where $\hat{e_i}$ is the $i$-th coordinate vector \cite{Nihat05}. The exponential map reduces to the exponential family at the barycenter $b = (\frac{1}{n}, \ldots, \frac{1}{n})$,
\[ exp(b,v) = \sum_{i}{ \frac{\frac{1}{n}e^{v_i}}{ \sum_{j}{\frac{1}{n} e^{v_j}} } \hat{e_i} } = \frac{1}{\sum_{j}{ e^{v_j}} }(e^{v_1}, \ldots, e^{v_n}).\]

% Define an \emph{exponential family} to be a collection of distributions of the form
% \[ p(x; \theta) = \text{exp}\left(C(x) + \sum_{i}{ \theta_i F_i(x) - \psi(\theta)} \right),\]
% for functions $F_i$, $C$, and $\psi$ and parameter vector $\theta$. These are maximal entropy distributions with respect to constraints of the form $E\left[F_i(x)\right] = \lambda_i$ and can be derived with Lagrange multipliers. Exponential families are maximal entropy distributions \cite{Naudts08}.

The solutions of the replicator equation can be realized as exponential families \cite{Karev09, Nihat05, Akin82}. Let $x_i = \exp (v_i - G)$ with $\dot{v_i} = f_i(x)$ and $G(x)$ a normalization constant to ensure that the distribution sums to one. From the fact that $\sum_{i}{x_i} = 1$, $0 = \sum_{i}{\dot{x_i}}$ and so
\begin{align*}
0 = \sum_{i}{\dot{x}_i} &= \sum_{i}{\exp (v_i(x) - G(x)) (\dot{v}_i(x) - \dot{G}(x))}\\
&= \sum_{i}{x_i (\dot{v}_i(x) - \dot{G}(x))} = \sum_{i}{(x_i f_i(x))} - \dot{G}(x)\\
&= \bar{f}(x) - \dot{G}(x)
\end{align*}
Hence $\dot{G} = \bar{f}(x)$. Now $x_i$ satisfies
\[ \dot{x_i} = \exp (v_i(x) - G(x)) (\dot{v_i}(x) - \dot{G}(x)) = x_i (f_i(x) - \bar{f}(x)), \]
which is the replicator equation. In the case of a log-linear fitness landscape, explicit solutions can be derived \cite{Nihat05}. In this case, the equation for the variable $v$ can be reduced to a linear differential equation, which can be solved with eigenvalue methods.

% Entropy is a measure of the extent to which the distribution has ``spread out'' over the landscape. Natural selection acts to fit the available niches in a fitness landscape, arriving at the maximal entropy distribution allowed by the constraints of the landscape. In the absence of variation within the fitness landscape, the replicator dynamic is stable. In fact, if $f_i(x) = c$ for all $i$ and all $x$ then any $x$ is stationary (and a Nash equilibrium), and the solution with maximal entropy distribution is the uniform distribution, which follows directly from the definition of the exponential family. In the case of a variable fitness landscape, natural selection realigns the population distribution to fill out the landscape if not at equilibrium, at each point taking the maximal entropy distribution available within the constraints of the values of the $f_i(x)$.

% \subsection{Representations}
% Several coordinate systems are used in information geometry. One system, in particular, is interesting in the context of studying natural selection. Consider the exponential representation of $P(X)$ via the embedding $p \mapsto \log p$ into the set of functions ${\R}^X = \{ g : X \to \R \}$. The tangent space at the point $p$ is characterized by functions with mean zero taken at $p$, i.e. $E_{p}[g] = 0$. Parallel translation of a vector $v$ from the tangent space at $p$ to a vector $v'$ the tangent space at $q$ is given by $v' = v - E_{q}[v]$. The later equation 

\subsection{Denormalization}

Information geometry defines the \emph{denormalized manifold} $\tilde{P}(X) = \{ \tau p | \tau \in \R^+, p \in P(X)\}$, which can be thought of as non-normalized discrete probability distributions. As with $P(X)$, $\tilde{P}(X)$ has an information metric. The denormalized manifold embeds into the reals as $\R^{n}_{+}$, with the denormalized information metric induced by the the metric given by Shahshahani, where the mapping back onto $P(X)$ realizes $\tau$ as $|x|$. In coordinates, the metric is given by \[\tilde{g}_{ij}(x) = \tau g_{ij}(x) = \frac{\tau}{x_i}\delta_{ij},\]
% where $\tau$ corresponds to $|x|$ in the metric given by Shahshahani. This information is encoded in the following commutative diagram.
% 
% \[ \begin{CD}
% \tilde{P}(X) @> >> \R^{n}_{+} \\
% @VV ||\cdot|| V @VV ||\cdot|| V\\
% P(X) @> >> \Delta^{n}
% \end{CD} \]
% 
% % commutative diagram for \tau maps?
% Clearly, many other normalizations are possible, such as the well-known $p$-norms $||\cdot||$, as well as another type of normalization given by an escort function in the next chapter. Discussions of alternative norms does not appear to be in the literature.

Akin uses the metric \[ g_{ij}(x) = \frac{1}{x_i}\delta_{ij} \] on $\R^n_{+}$ rather than the metric \[ g_{ij}(x) = \frac{|x|}{x_i}\delta_{ij} \]
given by Shahshahani\cite{Shahshahani79, Akin82}. Both metrics restrict to the same metric given by Shahshahani on the simplex. From the point of view of information geometry, the metric given by Shahshahani is the natural choice. The choice affects the form of the gradient on $\R^n_{+}$, which is in the case of Akin's metric is the Lotka-Volterra predator-prey equation.

\subsection{The Lotka-Volterra Equations and the Replicator Equation}

The Lotka-Volterra equations
\begin{equation}\label{lotka_volterra}
\dot{x}_i = x_i f_i(x)
\end{equation}
descend from $\R^{n}_{+}$, through a normalization map onto the simplex, to a replicator equation with an altered landscape. To see this, let $|x| = x_1 + \cdots + x_n$, $\dot{x}_i = x_i f_i(x)$ and $y_i = \frac{x_i}{|x|}$. Rearrange to $|x|y_i = x_i$ and note that $\frac{d}{dt}{|x|} = \sum_{i}{\dot{x}_i} = \sum_{i}{x_i f_i(x)} = x \cdot f(x)$. By the product rule, $\frac{d}{dt}{|x|}y_i + |x|\dot{y}_i = \dot{x}_i$ and so
\begin{align*}
\frac{d}{dt}{y_i} &= \frac{\dot{x}_i - \frac{d}{dt}{|x|}\dot{y}_i }{|x|}\\
&= \frac{x_i f_i(x) - x \cdot f(x) y_i}{|x|}\\
&= y_i ( f_i(x) - y \cdot f(x) )\\
&= y_i ( g_i(y) - y \cdot g(y) ),
\end{align*}
where $g_i(y) = f_i(x)$ is an alteration of the fitness landscape.

The Lotka-Volterra equations are the gradient flow with respect to the metric given by Akin on $\R^n_{+}$\cite{Hofbauer98}. The gradient of the metric given by Shahshahani differs by a factor of $|x|$:
\begin{equation}\label{shifted_lotka_volterra}
\dot{x}_i = \frac{x_i}{|x|} f_i(x),
\end{equation}

This system is transformable to Equation \ref{lotka_volterra} after a change of velocity eliminating the scalar function $B(x) = \frac{1}{|x|}$ because $B$ is strictly positive on $\R^{n}_{+}$. Equation \ref{shifted_lotka_volterra} transforms to a replicator equation via the normalization map \cite{Shahshahani79}.

% Indeed, consider the normalization map $x \to \frac{x}{|x|}$. The Jacobian is given by $J_{ii} = 1 - \frac{x_i}{|x|}$ and $J_{ij} = - \frac{x_i}{|x|}$ for $i \neq j$. The product $J \dot{x}$ gives the replicator equation in the variables $y_i = \frac{x_i}{|x|}$.

The Lotka-Volterra equations can be interpreted as a gradient of the denormalized Fisher information metric, in the case that $f$ is an Euclidean metric, in analogy to the replicator equation. This allows denormalized analogues of earlier results, such as the following, which is true for the denormalized version of the Lotka-Volterra equation.

\begin{theorem}\label{lv_lyapunov}
Let $\ess{x}$ in $\R^n_{+}$ be such that \[\frac{\ess{x} \cdot f(x)}{|\ess{x}|} > \frac{x \cdot f(x)}{|x|}\]
in some neighborhood of $\ess{x}$ (a denormalized ESS). Suppose that the trajectory of $\dot{x}_i = \frac{x_i}{|x|} f_i(x)$ lies in a set that contains no point parallel to $\ess{x}$.
% and $\ess{x} \cdot f(\ess{x}) = 0$, which implies that $\ess{x}$ is a rest point of \ref{lotka_volterra}.
Then the denormalized Kullback-Liebler divergence 
\[ D_{KL}\left(\frac{\ess{x}}{|\ess{x}|} || \frac{x}{|x|} \right)\] is a local Lyapunov function for Equation \ref{shifted_lotka_volterra}.
\end{theorem}
\begin{proof}
The divergence is minimal (and equal to zero) when $x = c \ess{x}$ for some constant $c$. Hence if the line through the origin and the point $\ess{x}$ intersects the the trajectory at most once, the divergence is zero if and only if $\ess{x} = x$. The time derivative is
\begin{align*}
\frac{d}{dt}\left[D_{KL}\left(\frac{\ess{x}}{|\ess{x}|} || \frac{x}{|x|}\right)\right] &= 0 -\frac{d}{dt}\left[ \sum_{i}{\frac{\ess{x}_i}{|\ess{x}|} (\log{x_i} - \log{|x|})} \right]\\
&= -\sum_{i}{\frac{\ess{x}_i}{|\ess{x}|} \frac{\dot{x}_i}{x_i}} + \sum_{i}{\frac{\ess{x}_i}{|\ess{x}|} \frac{\dot{|x|}}{|x|}} \\
&= -\sum_{i}{\frac{\ess{x}_i}{|\ess{x}|} \frac{1}{x_i} \frac{x_i}{|x|} f_i(x) } + \frac{x \cdot f(x) }{|x|^2} \\
&= -\frac{1}{|x|} \sum_{i}{\frac{\ess{x}_i}{|\ess{x}|} f_i(x)}  + \frac{1}{|x|^2} x \cdot f(x)\\
&= -\frac{1}{|x|}\left( \frac{\ess{x} \cdot f(x)}{|\ess{x}|} - \frac{x \cdot f(x)}{|x|} \right) < 0.
\end{align*}
\end{proof}

\section{Informatics of Multiple Population Replicator Dynamics}

The information-theoretic approach easily extends to multiple population replicator equations such as bimatrix games. As before, the potential information plays a crucial role. It is the sum of the potential informations of all populations that plays the role of the Lyapunov function and gives rise to the geometry. It suffices to discuss the two population case as it is clear that the results extend inductively to finitely-many populations.

\subsection{Two Populations}
Consider two categorical distributions $p = (p_1, \ldots, p_n)$ and $q = (q_{n+1}, \ldots, q_{n+m})$ with fitness landscapes $f(p,q) = (f_1(p, q), \ldots, f_n(p, q))$ and $g(p, q) = (g_{n+1}(p,q), \ldots, g_{n+m}(p,q))$. Define the coupled replicator system

\begin{align*}
\dot{p_i} &= p_i(f_i(p,q) -\ex{p}{f(p,q)}) \\
\dot{q_j} &= q_j(g_j(p,q) -\ex{q}{g(p,q)})
\end{align*}
where $i$ runs from 1 to $n$ and $j$ runs from $n+1$ to $n+m$. Note carefully that the expected values are taken with each distribution respectively.

This system is the gradient flow of the Riemannian metric defined on the interior of $\Delta^n \times \Delta^{m}$ given by
\[G_{i,j}(p, q) = \begin{cases}
                  \frac{1}{p_i} & \text{if $i = j \leq n$ } \\
                  \frac{1}{q_i} & \text{if $i = j > n$ } \\
		  0 & \text{else}
                 \end{cases}\]
That is, the matrix is the direct sum matrix of the usual metric for each equation. As in the single population case, we can use potential information to form a Lyapunov function for the system. Given states $\hat{p}$ and $\hat{q}$, let $L$ be the sum of the potential information of each categorical distribution. That is,
\begin{align*}
L &= D_{p}(\hat{p}, p) + D_{q}(\hat{q}, q)\\
&= \sum_{i}{\hat{p}_i \log{\hat{p}_i}} - \sum_{i}{\hat{p}_i \log{p_i}} + \sum_{j}{\hat{q}_j \log{\hat{q}_j}} - \sum_{j}{\hat{q}_j \log{q_j}}
\end{align*}

The metric can be obtained as the localization of the sum of the divergence functions. All the usual calculations follow from the fact that the system is a gradient, e.g. Fisher's Fundamental Theorem.

\subsection{Potential Information is a Lyapunov Function}

Recall that $\hat{p}$ is an ESS in the single population case if $\hat{p} \cdot f(p) > p \cdot f(p)$ for all $p$ in a neighborhood of $\hat{p}$.

\begin{theorem}
If $\hat{p}$ and $\hat{q}$ are ESS for each system respectively then $L$ is a Lyapunov function for the coupled system.
\end{theorem}
\begin{proof}
A straight-forward computation shows that (up to a negative)
\[ \dot{L} = \hat{p} \cdot f(p, q) - p \cdot f(p, q) + \hat{q} \cdot g(p, q) - q \cdot g(p, q).\]
$L$ is positive everywhere and has minimum at $(\hat{p}, \hat{q})$. Since $\hat{p}$ and $\hat{q}$ are ESS, $\dot{L} < 0$, so $L$ is a local Lyapunov function.
\end{proof}

Notice that the hypothesis that both $\hat{p}$ and $\hat{q}$ are ESS is too strong. Indeed, all that is required is that
\[ \hat{p} \cdot f(p,q) + \hat{q} \cdot g(p,q) > p \cdot f(p, q) + q \cdot g(p,q). \]
Call this condition a \emph{coupled ESS} (as well as its obvious higher dimensional analogs) and note that any ESS is a coupled ESS. Then $L$ is a Lyapunov for the system if and only if ($\hat{p}$ and $\hat{q}$) is a coupled ESS for the two population system.

% \subsubsection{Questions}
% 
% \begin{enumerate}
%   \item Is every coupled ESS necessarily an ESS?
%   \item If not, is an interior coupled ESS unique? asymptotically stable?
% \end{enumerate}

\subsection{Solutions}
We can again show that the solutions are exponential families. Let $\dot{v} = f(p, q)$ and $\dot{w} = g(p, q)$. Let $\dot{N} = \ex{p}{f(p,q)}$ and $\dot{M} = \ex{q}{g(p,q)}$. Then $p_i = \exp{(v_i - N)}$ and $q_j = \exp{(w_j - M)}$ is a solution to the coupled system. Indeed,
\[ \dot{p}_i = \exp{(v_i - N)} (\dot{v}_i - \dot{N}) = p_i( f_i(p, q) - \ex{p}{f(p, q)}),\]
and similarly for $q_j$.

\subsection{Multiple Populations}

The above generalizes by induction to show that for a coupled system of multiple interacting populations, the sum of the respective potential informations gives a Lyapunov function for a coupled ESS.

\section{Discussion}

The Shahshahani geometry can be interpreted within the framework of information theory as the information geometry of the simplex. This explains the origin of several quantities in evolutionary game theory including the Shahshahani metric and the Kullback-Liebler information divergence. An important feature of the approach is that the information-geometric reasoning extends to the Lotka-Volterra equation and the multiple population replicator equation easily within the framework. Additionally, the replicator dynamic arises intuitively from purely mathematical and statistical concepts such as Fisher information. This shows that the replicator equation models the information dynamics of natural selection.

\bibliography{ref}
\bibliographystyle{plain}

\end{document}